\RequirePackage{luatex85}

\documentclass[11pt,a4paper]{article}
\usepackage{iftex}

\ifPDFTeX
\usepackage[utf8]{inputenc}
\usepackage[T1]{fontenc}
\fi

\usepackage{amsmath}
\usepackage{amssymb}
\usepackage{amsthm}
\usepackage{mathtools}
\usepackage{thmtools}

\ifPDFTeX
\usepackage{libertine}
\usepackage[libertine]{newtxmath} 
\usepackage[scaled=0.96]{zi4} 
\else
\usepackage{fontspec}
\usepackage{unicode-math}
\fi

\usepackage[DIV=11]{typearea} 

\usepackage{microtype}

\usepackage{enumitem}
\usepackage{tikz}

\usepackage[
	style=alphabetic,
	backref=true,
	doi=true,
	url=true,
	maxcitenames=3,
	mincitenames=3,
	maxbibnames=10,
	minbibnames=10,
	backend=biber
]{biblatex}

\usepackage[ocgcolorlinks]{hyperref} 
\usepackage{cleveref}

\makeatletter
\g@addto@macro\bfseries{\boldmath}
\makeatother

\makeatletter
\g@addto@macro\mdseries{\unboldmath}
\g@addto@macro\normalfont{\unboldmath}
\g@addto@macro\rmfamily{\unboldmath}
\g@addto@macro\upshape{\unboldmath}
\makeatother

\DeclareCiteCommand{\citem}
    {}
    {\mkbibbrackets{\bibhyperref{\usebibmacro{postnote}}}}
    {\multicitedelim}
    {}

\renewcommand*{\multicitedelim}{\addcomma\space}

\AtEveryCitekey{\clearfield{note}}

\makeatletter
\AtBeginDocument{%
    \newlength{\temp@x}%
    \newlength{\temp@y}%
    \newlength{\temp@w}%
    \newlength{\temp@h}%
    \def\my@coords#1#2#3#4{%
      \setlength{\temp@x}{#1}%
      \setlength{\temp@y}{#2}%
      \setlength{\temp@w}{#3}%
      \setlength{\temp@h}{#4}%
      \adjustlengths{}%
      \my@pdfliteral{\strip@pt\temp@x\space\strip@pt\temp@y\space\strip@pt\temp@w\space\strip@pt\temp@h\space re}}%
    \ifpdf
      \typeout{In PDF mode}%
      \def\my@pdfliteral#1{\pdfliteral page{#1}}
      \def\adjustlengths{}%
    \fi
    \ifxetex
      \def\my@pdfliteral #1{}
      \def\adjustlengths{\setlength{\temp@h}{-\temp@h}\addtolength{\temp@y}{1in}\addtolength{\temp@x}{-1in}}%
    \fi%
    \def\Hy@colorlink#1{%
      \begingroup
        \ifHy@ocgcolorlinks
          \def\Hy@ocgcolor{#1}%
          \my@pdfliteral{q}%
          \my@pdfliteral{7 Tr}
        \else
          \HyColor@UseColor#1%
        \fi
    }%
    \def\Hy@endcolorlink{%
      \ifHy@ocgcolorlinks%
        \my@pdfliteral{/OC/OCPrint BDC}%
        \my@coords{0pt}{0pt}{\pdfpagewidth}{\pdfpageheight}%
        \my@pdfliteral{F}
        \my@pdfliteral{EMC/OC/OCView BDC}%
        \begingroup%
          \expandafter\HyColor@UseColor\Hy@ocgcolor%
          \my@coords{0pt}{0pt}{\pdfpagewidth}{\pdfpageheight}%
          \my@pdfliteral{F}
        \endgroup%
        \my@pdfliteral{EMC}%
        \my@pdfliteral{0 Tr}
        \my@pdfliteral{Q}%
      \fi
      \endgroup
    }%
}
\makeatother

\colorlet{DarkRed}{red!50!black}
\colorlet{DarkGreen}{green!50!black}
\colorlet{DarkBlue}{blue!50!black}

\hypersetup{
	linkcolor = DarkRed,
	citecolor = DarkGreen,
	urlcolor = DarkBlue,
	bookmarksnumbered = true,
	linktocpage = true
}

\allowdisplaybreaks[1]

\declaretheorem[numberwithin=section]{theorem}
\declaretheorem[numberlike=theorem]{lemma}

\newcommand{\dist}{d}
\newcommand{\poly}{\operatorname{poly}}
\newcommand{\polylog}{\operatorname{polylog}}

\title{A Faster Distributed Single-Source Shortest Paths Algorithm\thanks{Presented at the \emph{59th Annual IEEE Symposium on Foundations of Computer Science (FOCS 2018)}.}}

\author{
Sebastian Forster\thanks{Department of Computer Sciences, University of Salzburg, Austria. Work partially done while at University of Vienna, Austria. This author previously published under the name Sebastian Krinninger.}
\and Danupon Nanongkai\thanks{EECS, KTH Royal Institute of Technology, Sweden.}
}

\date{}

\hypersetup{
	pdftitle = {A Faster Distributed Single-Source Shortest Paths Algorithm},
	pdfauthor = {Sebastian Forster, Danupon Nanongkai}
}

\begin{document}
\maketitle
\begin{abstract}
We devise new algorithms for the single-source shortest paths (SSSP) problem with non-negative edge weights in the CONGEST model of distributed computing.
While close-to-optimal solutions, in terms of the number of rounds spent by the algorithm, have recently been developed for computing SSSP approximately, the fastest known exact algorithms are still far away from matching the lower bound of $ \tilde \Omega (\sqrt{n} + D) $ rounds by Peleg and Rubinovich~\citem[SICOMP'00]{PelegR00}, where $ n $~is the number of nodes in the network and $ D $ is its diameter.
The state of the art is Elkin's randomized algorithm~\citem[STOC'17]{Elkin17} that performs $ \tilde O(n^{2/3} D^{1/3} + n^{5/6}) $ rounds.
We significantly improve upon this upper bound with our two new randomized algorithms for polynomially bounded integer edge weights, the first performing $ \tilde O (\sqrt{n D}) $ rounds and the second performing $ \tilde O (\sqrt{n} D^{1/4} + n^{3/5} + D) $ rounds. Our bounds also compare favorably to the independent result by Ghaffari and Li~\citem[STOC'18]{GhaffariL18}.
As side results, we obtain a $(1+\epsilon)$-approximation $\tilde O ((\sqrt{n} D^{1/4} + D) / \epsilon) $-round algorithm for \emph{directed} SSSP and  a new work/depth trade-off for exact SSSP on directed graphs in the PRAM model.

\end{abstract}

\section{Introduction}

In this paper, we consider the fundamental problem of computing single-source shortest paths (SSSP) in the CONGEST model~\cite{Peleg00} of distributed computing.
The CONGEST model is one of the major message-passing models studied in the literature and is characterized by synchronized communication in a network via non-faulty bounded-bandwidth links.
Such a communication network is modeled by a graph in which the nodes correspond to the processors in the network and the edges correspond to the communication links between the processors.

In unweighted graphs, the SSSP problem can be solved in $ O (D) $ communication rounds by performing breadth first search, where $ D $~is the diameter of the graph.
In weighted graphs, where we assume that the edge weights do \emph{not} represent the communication speed, a straightforward distributed variant of the Bellman-Ford algorithm~\cite{Ford56,Bellman58,Moore59} computes SSSP in $ O (n) $ rounds, where $ n $ is the number of nodes in the network.
Moreover, Peleg and Rubinovich~\cite{PelegR00} showed a lower bound of $ \tilde \Omega (\sqrt{n} + D) $ for this problem\footnote{In this paper, we use $ \tilde O (\cdot) $-, $ \tilde \Omega (\cdot) $-, $ \tilde \Theta (\cdot) $-, and $ \tilde o (\cdot) $-notation to suppress factors that are polylogarithmic in $ n $.}, where $ D $ is the \emph{unweighted} diameter of the graph, when viewed as the underlying communication network. ($D$ is often called {\em hop diameter}.)
The past few years have witnessed many improved upper bounds when {\em approximate solutions} are allowed (e.g. ~\cite{LenzenPS13,Nanongkai14,HenzingerKN16,ElkinN16,BeckerKKL17,HaeuplerL18,HaeuplerHW18}), culminating in the close-to-optimal~\cite{DasSarmaHKKNPPW12,PelegR00,Elkin06,KorKP13,ElkinKNP14} approximation scheme of Becker et al.~\cite{BeckerKKL17}.
For {\em exact} algorithms, the $O(n)$ bound from the Bellman-Ford algorithm was the state of the art for a long time until in a recent breakthrough, Elkin~\cite{Elkin17} proved an upper bound that is sublinear in~$ n $ for undirected graphs with non-negative edge weights; he obtained a randomized algorithm that performs $ O((n \log n)^{5/6}) $ rounds for $ D = O (\sqrt{n \log n}) $, and $ O (D^{1/3} (n \log n)^{2/3}) $ rounds for larger~$ D $.

\subsection{Our Results}

Our main result is a significant improvement upon Elkin's upper bound for polynomially bounded integer edge weights\footnote{In general, our running times scale with the number of bits needed to represent edge weights.} and $ D = \tilde o (n) $. We devise two randomized (Las Vegas) algorithms, the first performing $ \tilde O (\sqrt{n D}) $ rounds and the second performing $ \tilde O (\sqrt{n} D^{1/4} + n^{3/5} + D) $ rounds; both bounds hold with high probability and in expectation. Our first bound matches the $ \tilde \Omega (\sqrt{n} + D)$ lower bound of Peleg and Rubinovich~\cite{PelegR00} up to polylogarithmic factors when $ D $ is polylogarithmic in $ n $. 

Our algorithms in fact work in a more restricted model called {\em Broadcast CONGEST}, where a node must send the same message to every neighbor in each round. 
The bound holds when edge weights are in $\{0, 1, \ldots, \poly(n)\}$ as typically assumed in the literature. Note that our bounds hold even when edges are {\em directed}.
The directed case refers to when edges have directions, but the communication is always bi-directional. A question about this case was raised in \cite{Nanongkai14}. The previous bound by Elkin does not work for this case. 
Note that in general our bounds have to be multiplied by the number of bits needed to represent edge weights (which is $O(\log n)$ under a typical assumption), while Elkin's bound does not.

\paragraph{Independent Work by Ghaffari and Li \cite{GhaffariL18}.} Recently in STOC 2018, and independently from our results, Ghaffari and Li presented a distributed algorithm for exact SSSP with time complexity $\tilde O(n^{3/4}D^{1/4})$ and another algorithm with complexity $\tilde O(n^{3/4+o(1)}+\min\{n^{3/4}D^{1/6}, n^{6/7}\}+ D)$. 
Our bounds compare favorably with theirs in the entire range of parameters.
Like our bounds, their bounds hold when edges are directed. Their bounds also depend on the number of bits needed to represent edge weights.

\paragraph{Additional Results.} 

Using techniques developed in this paper, we obtain two additional results, which might be interesting for putting our main result and our techniques into context.
The first is a $(1+\epsilon)$-approximation $\tilde O ((\sqrt{n} D^{1/4} + D) / \epsilon) $-round algorithm for \emph{directed} SSSP.
This algorithm is Monte Carlo, meaning that its approximation guarantee holds with high probability, and works in the Broadcast CONGEST model. 
Previously such a bound was known only for the special case of {\em single-source reachability} \cite{GhaffariU15}, where we want to know whether there is a directed path from the source node to each node. None of previous approximation algorithms for SSSP  (e.g. \cite{LenzenPS13,Nanongkai14,HenzingerKN16,ElkinN16,BeckerKKL17}) can handle this case.

Our second result is a new work/depth trade-off for exact SSSP on directed graphs in the PRAM model.
We provide algorithm with $ \tilde O (m t + m n / t + n^3/t^3) $ work and $ \tilde O (t) $ depth (for any fixed $ 1 \leq t \leq n $).
The parallel SSSP problem has received considerable attention in the literature~\cite{Spencer97,KleinS97,Cohen97,BrodalTZ98,ShiS99,Cohen00,MeyerS03,MillerPVX15,Blelloch0ST16}, but we are aware of only two relevant results for exact SSSP in directed graphs.
First, the algorithm of Spencer~\cite{Spencer97} has $ \tilde O (n^3 / t^2 + m) $ work and $ \tilde O (t) $ depth (for any $ 1 \leq t \leq n $).
Second, the algorithm of Klein and Subramanian has $ \tilde O (m \sqrt{n}) $ work and $ \tilde O (\sqrt{n}) $ depth, but does not give any trade-offs.
Thus, our algorithm gives a better trade-off than the state of the art when $ t = \tilde o (\sqrt{n}) $ and in this case it has $ \tilde O (m n / t + n^3/t^3) $ work and $ \tilde O (t) $ depth.
In this range, the trade-off is (up to polylogarithmic factors) the same as the one of Shi and Spencer~\cite{ShiS99} for undirected graphs.\footnote{The algorithm of Shi and Spencer uses a hop set construction that is tailored to undirected graphs; the same hop set construction has later also been used in the distributed setting~\cite{Nanongkai14,Elkin17}.}

\subsection{Further Related Work}

Closely related are {\em All-Pairs Shortest Paths (APSP)} and {\em $k$-Source Shortest Paths} ($k$-SSP) problems. The current best upper bound for APSP is $\tilde O(n^{5/4})$ \cite{HuangNS17}. The situation is more complicated for $k$-SSP. We refer to \cite{FrischknechtHW12,PelegRT12,HolzerW12,LenzenP13,Elkin17,HuangNS17,GhaffariL18,AgarwalRKP18} and references therein for details.

The exact shortest paths problem has also received attention in the Congested Clique model, which is the special case of the CONGEST model where the network is a clique.
Nanongkai obtained a SSSP algorithm that performs $ \tilde O (\sqrt{n}) $ rounds in this model and Censor-Hillel et al.~\cite{CensorHillelKKLPS15} obtained an algorithm that performs $ \tilde O (n^{1/3}) $ rounds, solving the more general APSP problem.

There are many other graph problems studied in the CONGEST model, such as minimum spanning tree, minimum cut, and maximum flow; see, e.g., \cite{Gafni85,ChinT85,Awerbuch87,GarayKP98,KuttenP98,GhaffariK13,Ghaffari14,NanongkaiS14_disc,GhaffariKKLP18,PanduranganRS-STOC17,Elkin-podc17-mst,Bar-YehudaCGS17}.

\subsection{Overview of Techniques}

The starting point of our algorithms is the classic {\em scaling framework} which was heavily used in the sequential setting since the 80s (e.g. \cite{Gabow85,GabowT89,Goldberg95-soda93}). In the context of distributed shortest paths, this framework was used first in the algorithm of  Huang~et~al.~\cite{HuangNS17} for APSP. This algorithm follows Gabow's \emph{bit-wise} scaling technique~\cite{Gabow85}, where edge weights are considered one bit at a time so that distances between nodes can be bounded from above. This approach was later taken by Ghaffari and Li~\cite{GhaffariL18} for computing SSSP. 
Our starting point is rather different from Huang~et~al.~\cite{HuangNS17} and Ghaffari and Li~\cite{GhaffariL18} in that we will follow \emph{recursive} scaling, Gabow's second scaling technique~\cite{Gabow85}, and its extension by Klein and Subramanian~\cite{KleinS97}, who applied it for SSSP in the PRAM model.
At a high level, this approach uses a very general reduction to extend certain {\em approximate} SSSP algorithms to {\em exact} SSSP algorithms. 
Intuitively, this allows us to borrow some existing tools developed for approximation algorithms. 

The challenge here lies in two special aspects of this reduction.
First, it is an inherent necessity of the reduction that the approximation algorithm works on directed graphs and thus our exact SSSP algorithm also works on directed graphs -- assuming that the underlying communication network is undirected. Because of this, it is not a surprise that we achieve a new approximation algorithm that can handle the directed case as an additional result.

Second, and more importantly, the reduction will not work for any arbitrary approximate SSSP algorithm, as it requires the distance estimates returned by the algorithm to satisfy an additional condition about `dominating' the edge weights of the input graph.
Klein and Subramanian can ensure this property by (1) computing a directed hop set to effectively reduce the approximate shortest path diameter of the graph, (2) adding the edges of the hop set to the input graph, and (3) obtaining the distance estimate from performing a ``small'' number iterations of the Bellman-Ford algorithm, exploiting the reduced approximate shortest path diameter.
This approach cannot directly be transferred to the CONGEST model because the additional edges of the hop set can only be ``simulated'' by the nodes in the network, and the standard way of doing so will not yield a small enough guarantee on the number of rounds spent in the final Bellman-Ford step.
We circumvent this problem by using a different algorithm design instead, where we (1) compute a skeleton graph consisting of a few ``important'' nodes, including the source node, and edges between the skeleton nodes with weights equal to approximate pairwise distances, (2) solve recursively on the skeleton graph, and (3) run a ``small'' number of iterations of the Bellman-Ford algorithm incorporating the distances from the source computed in step~(2) to obtain distance estimates with the domination property.
This may sound a bit circular at first because we again face the problem of having to compute exact SSSP on the skeleton graph.
However, computation on the smaller skeleton graph is now simulated by performing global broadcasts in the whole network.
We can thus benefit from the fact that we can use slightly different algorithm design techniques, such as ``adding'' edges to the graph without immediately increasing the number of rounds, just like in the PRAM model.

To this end, we note that although our algorithm follows the same general framework as that by Klein and Subramanian~\cite{KleinS97}, it is not simply a mere ``translation'' of their PRAM algorithm to the CONGEST model. It in fact leads to a new work/depth trade-off in the PRAM model that could not be achieved with the algorithm of Klein and Subramanian.

\subsection{Organization}

The rest of this paper is organized as follows:
In Section~\ref{sec:preliminaries} we introduce all definitions necessary for our main results and review existing tools from the literature.
In Section~\ref{sec:auxiliary algorithm} we schematically develop the central ``auxiliary'' algorithm that we need to obtain our results.
We prove its correctness, but do not yet analyze its model-specific complexity.
In Sections~\ref{sec:first distributed algorithm} and~\ref{sec:second distributed algorithm} we obtain our two main results by providing two slightly different implementations of the auxiliary algorithm together with their complexity analyses.
Finally, in Section~\ref{sec:additional_results} we sketch how our techniques imply additional results for approximate SSSP in the Broadcast CONGEST model and exact SSSP in the PRAM model.

\section{Preliminaries}\label{sec:preliminaries}

\subsection{CONGEST Model and Problem Formulation}

The \emph{CONGEST model}~\cite{Peleg00} is a synchronous message-passing model with non-faulty bounded-bandwidth links.
More formally, consider a communication network of processors modeled by an undirected graph $ N = (V, L) $ where each node of $ V $ models a processor and each pair of nodes $ \{ u, v \} \in L \subseteq \binom{V}{2} $ indicates a bidirectional communication link between the processors corresponding to $ u $ and $ v $, respectively.
In the remainder of this paper, we identify nodes and processors.
In the CONGEST model, the initial knowledge about $ N $ is distributed among the nodes: every node has a unique identifier of size $ O (\log{n}) $ (where $ n = |V| $) initially known only to itself and its neighbors, i.e., the nodes to which it has direct communication links.
The nodes communicate in synchronous rounds: At the beginning of each round, every node may send to each of its neighbors a message of size $ B =\Theta (\log{n}) $ and subsequently receive the messages sent by its neighbors.
Before the next round begins, each node may perform internal computation based on all messages it has received so far and its local knowledge of the network.
The complexity of an algorithm is usually measured in the number of rounds and the total number of messages sent until the algorithm terminates, whereas internal computation is considered free.
Typically, the asymptotic complexity is expressed in terms of $ n := |V| $ and the diameter $ D $ of the unweighted communication network~$ N $, i.e., the maximum distance between any pair of nodes in $ N $.
The restriction of the bandwidth~$ B $ to $ \Theta (\log{n}) $ distinguishes the CONGEST model from the \emph{LOCAL model}~\cite{Linial92}.
The \emph{Broadcast CONGEST model} is a specialization of the CONGEST model with the additional constraint that, whenever a node sends a message in some round, it has to send (\emph{broadcast}) \emph{the same} message to all of its neighbors.
In general, we call a round in which, for every node, the message sent to the neighbors is the same, a \emph{broadcast round}.

In the single-source shortest paths (SSSP) problem, we are given a weighted directed graph $ G = (V, E, w) $ and a distinguished source node $ s \in V $ and the task is to compute the distance $ \dist_G (s, v) $ from $ s $ to $ v $ in $ G $ for every node $ v \in V $.
In the CONGEST model, we require that the nodes of $ G = (V, E, w) $ are equal to the nodes of the underlying communication network $ N = (V, L) $ and that for every directed edge $ (u, v) \in E $ there is an undirected communication link $ \{ u, v \} \in L $.
The input is distributed among the nodes of the network as follows: initially, every node knows (a) whether it is the source $ s $ or not and (b) its set of incoming and outgoing edges together with their weight.
The SSSP problem is solved as soon as every node knows its distance from $ s $ in $ G $.

In this paper, we consider the restriction of the SSSP problem to \emph{non-negative integer edge weights}, where we let $ W $ denote the maximum edge weight.
We work under the typical assumption that each edge weight can be encoded in a constant number of messages, i.e. $ B = \Omega(\log{W}) $, which implies that $ W $ is polynomial in $ n $.
Nevertheless, we do work out the dependence on $ W $ in our algorithms for the interested reader.
Observe, that in the CONGEST model an implicit shortest path tree rooted at~$ s $ can be constructed by performing an additional round: first every node sends its distance from $ s $ to all of its~ neighbors and then every node $ v $ determines one incoming neighbor $ u $ such that $ \dist_G (s, v) = \dist_G (s, u) + w (u, v) $ which will serve as the parent in the shortest path tree.

\subsection{Notation and Terminology}

Consider a weighted directed graph $ G = (V, E, w) $ with source node $ s \in V $, where $ V $ of size $ n := |V| $ is the set of nodes, $ E \subseteq V^2 $ of size $ m := |E| $ is the set of edges, and $ w : E \rightarrow \{ 0, 1, \dots, W \} $ assigns a non-negative integer weight to each edge.
For every path~$ \pi $ in~$ G $, denote the weight of $ \pi $ in $ G $ by $ w (\pi) := \sum_{e \in \pi} w (e) $.
For every pair of nodes $ u, v \in V $ denote the distance from $ u $ to $ v $ in $ G $, i.e., the weight of the shortest path from $ u $ to $ v $ in $ G $, by $ \dist_G (u, v) $.

In defining multiplicative distance approximation, one can consider both overestimation and underestimation of the true distance.
For technical reasons, we will refer to both types of distance estimates in our algorithms and thus use the following convention:
If $ \alpha \geq 1 $, then $ \tilde{\dist} (u, v) $ is an $ \alpha $-approximation of $ \dist_G (u, v) $ for a pair of nodes $ u, v, \in V $ if
\begin{equation*}
\dist_G (u, v) \leq \tilde{\dist} (u, v) \leq \alpha \cdot \dist_G (u, v)
\end{equation*}
and if $ \alpha < 1 $, then $ \hat{\dist} (u, v) $ is an $ \alpha $-approximation of $ \dist_G (u, v) $ if
\begin{equation*}
\alpha \cdot \dist_G (u, v) \leq \hat{\dist} (u, v) \leq \dist_G (u, v) \, .
\end{equation*}
Since the edge weights, and thus all pairwise distances, are integer, we assume without loss of generality that the distance estimates $ \tilde{\dist} (u, v) $ and $ \hat{\dist} (u, v) $ are integer (otherwise $ \lfloor \tilde{\dist} (u, v) \rfloor $ and $ \lceil \hat{\dist} (u, v) \rceil $, respectively, will serve as the desired $ \alpha $-approximations).

For every pair of nodes $ u, v \in V $ and every integer $ h \geq 1 $, we define the shortest $h$-hop path from $ u $ to~$ v $ to be the path of minimum weight among all paths from $ u $ to $ v $ with at most $ h $ edges.
The $h$-hop distance from $ u $ to $ v $ in $ G $, denoted by $ \dist_G^h (u, v) $ is the weight of the shortest $h$-hop path from $ u $ to $ v $.\footnote{Although the $h$-hop distances of all pairs of nodes do not necessarily induce a metric, the term `$h$-hop distance' is somewhat established in the literature.}

\subsection{Toolkit}

In the following, we review known tools that we use to design our algorithm.

\subsubsection{Reduction from Approximate SSSP to Exact SSSP}

We will apply the following reduction by Klein and Subramanian~\cite{KleinS97} for extending certain approximate SSSP algorithms to exact SSSP algorithms.
The reduction is based on Gabow's recursive scaling algorithm~\cite{Gabow85}.
Initially this reduction was formulated for the PRAM model, but it can be adapted to the CONGEST model in a straightforward way.
\begin{theorem}[Implicit in~\cite{KleinS97}]\label{thm:approximate to exact reduction distributed}
Assume there is an auxiliary algorithm $ \mathcal{A} $ in the CONGEST model that, given a directed graph $ G = (V, E, w_G) $ with non-negative integer edge weights and a source node~$ s $, computes a distance estimate $ \hat{\dist} (s, \cdot) $ such that
\begin{equation}
\tfrac{1}{2} \cdot \dist_G (s, v) \leq \hat{\dist} (s, v) \leq \dist_G (s, v) \label{eq:reduction approximation}
\end{equation}
for every node~$ v \in V $ and
\begin{equation}
\hat{\dist} (s, v) \leq \hat{\dist} (s, u) + w_G (u, v) \label{eq:reduction domination}
\end{equation}
for every edge $ (u, v) \in E $.
Then there is an exact SSSP algorithm $ \mathcal{B} $ in the CONGEST model that, given a directed graph $ G = (V, E, w_G) $ with non-negative integer weights in the range $ w_G (e) \in \{ 0, 1, \dots, W \} $ for every edge $ e \in E $ and a source node $ s $, makes $ O (\log{n W}) $ calls to algorithm~$ \mathcal{A} $ (on graphs $ G' = (V, E, w_{G'}) $ with non-negative integer weights in the range $ w_{G'} (e) \in \{0, \dots, O (n) \} $ for every edge $ e \in E $) and has an additional overhead of $ \tilde O (D \log{W}) $ broadcast rounds.
If $ \mathcal{A} $ is a randomized Monte Carlo algorithm that is correct with high probability, then $ \mathcal{B} $ is a Las Vegas algorithm whose bounds on the number of $ O (\log{n W}) $ calls to~$ \mathcal{A} $ and the additional overhead of $ \tilde O (D \log{W}) $ broadcast rounds, respectively, hold with high probability and in expectation.
\end{theorem}
We give a proof of Theorem~\ref{thm:approximate to exact reduction distributed} for completeness in Appendix~\ref{apx:reduction from approximate to exact}.

Our task of designing an exact SSSP algorithm thus reduces to designing an auxiliary approximate SSSP algorithm satisfying conditions \eqref{eq:reduction approximation}~and~\eqref{eq:reduction domination}.
Note that~\eqref{eq:reduction domination} is a special form of the \emph{triangle inequality} expressing that the distance estimates \emph{dominate} the edge weights of $ G $.
The reduction will fail if this additional constraint is not met, and thus we cannot use arbitrary distance estimates, which poses a challenge in designing the auxiliary algorithm.
As Klein and Subramanian~\cite{KleinS97} observed, one can for example ensure~\eqref{eq:reduction domination} by obtaining $ \hat{\dist} (s, \cdot) $ as the exact distances from $ s $ in a suitable supergraph $ G' = (V, E', w_{G'}) $ of $ G = (V, E, w_G) $ where $ E \subseteq E' $ and $ w_{G'} (e) \leq w_G (e) $ for every edge $ e \in E $.

\subsubsection{Computing Bounded Hop Distances}

We will use several primitives for computing (approximate) bounded-hop distances.
The well-known ``cornerstone'' for this task is the Bellman-Ford algorithm.

\begin{lemma}[Folklore]\label{lem:Bellman Ford CONGEST}
In the Broadcast CONGEST model, one can, given a weighted directed graph~$ G $ with source node $ s $ and an integer hop parameter $ h \geq 1 $, compute, for every node~$ v $, the value $ \dist_G^h (s, v) $ and make it known to $ v $ in $ O(h) $ rounds such that each node in total broadcasts at most $ O (h) $ messages to its neighbors by using a synchronized version of the Bellman-Ford algorithm
\end{lemma}

Note that the Bellman-Ford algorithm works for arbitrary weights, and in particular also for non-negative integer weights.

While the dependence on the number of rounds in the Bellman-Ford algorithm is optimal, there is some room for improvement in the number of messages.
Using a well-known weight-rounding technique~\cite{KleinS97,Cohen98,Zwick02,Bernstein09,Madry10,Bernstein13,Nanongkai14}, the number of messages can be reduced at the cost of introducing some approximation error.
Algorithmically, this technique amounts to performing a single-source shortest path computation up to bounded distance in an integer-weighted graph.
For this task, one can use a shortest path algorithm similar to breadth-first search that charges each node (at most) once for sending distance information to its neighbors.

\begin{lemma}[Implicit in~\cite{Nanongkai14}]\label{lem:approximate bounded hop distances CONGEST}
In the Broadcast CONGEST model, there is an algorithm that, given a directed graph $ G $ with non-negative integer edge weights, a fixed source node $ s $, and an integer hop parameter $ h \geq 1 $, computes, for every node~$ v $, a distance estimate $ \tilde d (s, v) $ ultimately known to $ v $ such that
\begin{equation*}
\dist_G (s, v) \leq \tilde{d} (s, v) \leq (1 + \epsilon) \dist_G^h (s, v)
\end{equation*}
in $ \tilde O (h \log{W} / \epsilon + D) $ rounds where each node in total broadcasts at most $ \tilde O (\log{W} ) $ messages to its neighbors.
\end{lemma}
Note that in~\cite{Nanongkai14} this algorithm was formulated for positive integer edge weights.
However, the result can be extended to non-negative integer edge weights with a simple reduction, that we provide for completeness in Appendix~\ref{apx:reduction from positive to non-negative weights} (Lemma~\ref{lem:positive to non-negative reduction distributed}).

This approximation algorithm can be leveraged to compute approximate $h$-hop distances from a set of $ k $ sources.
Note however that if we run $ k $ instances of the algorithm of Lemma~\ref{lem:approximate bounded hop distances CONGEST}, then in the worst-case it might happen that a single node might have to send different messages to its neighbors in all $ k $ instances of the algorithm simultaneously.
The naive approach of circumventing this type of congestion will blow up the number of rounds by a factor of $ k $.
Nanongkai~\cite{Nanongkai14} showed that the congestion can be avoided using a random delay technique.
\begin{lemma}[\cite{Nanongkai14}]\label{lem:approximate bounded hop distances multiple sources CONGEST}
In the Broadcast CONGEST model, there is a randomized algorithm that, given a directed graph $ G $ with non-negative integer edge weights, a fixed set of source nodes $ S $ of size $ k = |S| $, and an integer hop parameter $ h \geq 1 $, computes, for every source node $ s \in S $ and every node $ v $,  a distance estimate $ \tilde d (s, v) $ ultimately known to $ v $ such that, with high probability,
\begin{equation*}
\dist_G (s, v) \leq \tilde{d} (s, v) \leq (1 + \epsilon) \dist_G^h (s, v)
\end{equation*}
in $ \tilde O (h \log{W} / \epsilon + k + D) $ rounds.
\end{lemma}
Note that, in combination with the aforementioned rounding technique, one can alternatively use the deterministic source detection algorithm of Lenzen and Peleg~\cite{LenzenP13,LenzenP14} for this task.

\subsubsection{Randomized Hitting Set Construction}

Similar to previous randomized shortest path algorithms, we will use the fact that, for a specified parameter~$ h$, one can by a simple randomized process obtain a set $ C $ of size $ \tilde O (n/h) $ such that every shortest path consisting of $ h $ nodes contains a node of $ C $ with high probability, i.e., $ C $ is a \emph{hitting} set for the system of sets defined by the shortest paths with $ h $ nodes.
This technique was introduced to the design of graph algorithms by Ullman and Yannakakis~\cite{UllmanY91}.
A general lemma can be formulated as follows.
\begin{lemma}\label{lem:randomized hitting set}
Let $ c \geq 1 $, let $ U $ be a set of size $ s $, and let $ \mathcal{S} = \{ S_1, S_2, \ldots, S_k \} $ be a collection of sets over the universe $ U $ of size at least $ q $.
Let~$ T $ be a subset of $ U $ that was obtained by choosing each element of~$ U $ independently with probability $ p = \min (x / q, 1) $ where $ x = c \ln{(k s)} + 1 $.
Then, with high probability (whp), i.e., probability at least $ 1 - 1/s^c $, the following two properties hold:
\begin{enumerate}
\item For every $ 1 \leq i \leq k $, the set $ S_i $ contains an element of $ T $, i.e., $ S_i \cap T \neq \emptyset $.
\item $ |T| \leq 3 x s / q = O (c s \log{(k s)} / q) $.
\end{enumerate}
\end{lemma}
To apply the lemma for hitting shortest paths consisting of $ h $ nodes in a graph $ G = (V, E, w) $, set $ U = V $, $ q = h $, and for every pair of nodes $ u $ and $ v $ such that there is shortest path from~$ u $ to~$ v $ with exactly $ h $~nodes define a corresponding set $ S_i $ containing all the nodes on one of these shortest paths, resulting in $ k \leq n^2 $.

\section{Schematic Auxiliary Algorithm}\label{sec:auxiliary algorithm}

In the following, we present an \emph{auxiliary algorithm} that computes distance estimates satisfying the preconditions~\eqref{eq:reduction approximation} and~\eqref{eq:reduction domination} of Theorem~\ref{thm:approximate to exact reduction distributed} and can thus be extended to an exact SSSP algorithm using the recursive-scaling approach.
We formulate the auxiliary algorithm in a schematic manner and defer model-specific implementation details and the complexity analysis to later sections.
The auxiliary algorithm is parameterized by an integer parameter $ h $ whose relevance will only become clear in the complexity analysis.
Formally, the guarantees obtained in this section can be summarized as follows.
\begin{lemma}[Main Lemma]\label{lem:correctness auxiliary algorithm}
For any directed input graph $ G = (V, E, w_G) $ with non-negative integer edge weights and a fixed source node $ s $ and any integer $ 1 \leq h \leq n $, the auxiliary algorithm below consisting of \textbf{Steps 1--6} below computes, for every node $ v \in V $, a distance estimate $ \hat{\dist} (s, v) $ satisfying conditions~\eqref{eq:reduction approximation} and~\eqref{eq:reduction domination} of Theorem~\ref{thm:approximate to exact reduction distributed}.
\end{lemma}

The auxiliary algorithm proceeds as follows:
\begin{enumerate}[align=left,label=\textbf{Step \arabic*.},ref=\textbf{Step \arabic*}]
\item \label{step:construct hitting set}
Construct a set $ C \subseteq V $ of \emph{skeleton nodes} containing (a) the source node $ s $ and (b) additionally, for every pair of nodes $ u $ and $ v $ such that the shortest path from $ u $ to $ v $ in $ G $ consists of exactly $ \lceil h/2 \rceil $ nodes, at least one node on one of these shortest paths.

\item \label{step:compute approximate hop distances}
For each skeleton node $ x \in C $, compute $ 2 $-approximate $h$-hop distances from $ x $, i.e., distance estimates $ \tilde{\dist} (x, \cdot) $ such that
\begin{equation}
\dist_G (x, v) \leq \tilde{\dist} (x, v) \leq 2 \cdot \dist_G^h (x, v) \label{eq:upper and lower bound weighted BFS}
\end{equation}
for every node $ v \in V $.

\item \label{step:construct skeleton graph}
Construct the \emph{skeleton graph} $ H = (C, C^2, w_H) $ with edge weight $ w_H (x, y) = \tilde{\dist} (x, y) $ for every $ (x, y) \in C^2 $.

\item \label{step:exact SSSP on skeleton}
Compute distances from $ s $ on the skeleton graph~$ H $, i.e., $ \dist_H (s, x) $ for every skeleton node $ x \in C $.

\item \label{step:construct augmented graph}
Construct the \emph{augmented graph} $ G' = (V, E \cup \{ s \} \times C, w_{G'}) $ with the weight function given by
\begin{equation*}
w_{G'} (u, v) =
\begin{cases}
\dist_H (u, v) & \text{for every $ (u, v) \in (\{ s \} \times C) $} \setminus E \\
2 \cdot w_G (u, v) & \text{for every $ (u, v) \in E \setminus (\{ s \} \times C) $} \\
\min (\dist_H (u, v), 2 \cdot w_G (u, v)) & \text{for every $ (u, v) \in E \cap (\{ s \} \times C) $ \, .}
\end{cases}
\end{equation*}

\item \label{step:final Bellman-Ford}
Compute the $h$-hop distances from $ s $ in the augmented graph~$ G' $, i.e., $ \dist_{G'}^h (s, v) $ for every node $ v \in V $ and return $ \hat{\dist} (s, v) := \tfrac{1}{2} \cdot \dist_{G'}^h (s, v) $ for every node $ v \in V $.
\end{enumerate}

The correctness proof has two main parts.
We first argue (see Lemma~\ref{lem:approximation first auxiliary algorithm distributed}) that the distances in~$ G' $ are a $ 2 $-approximation of the distances in~$ G $.
Then we show (see Lemma~\ref{lem:triangle inequality first auxiliary algorithm distributed}) that the distance estimates computed by our algorithm, namely $ \hat{\dist} (s, v) := \tfrac{1}{2} \cdot \dist_{G'}^h (s, v) $ for every node $ v \in V $, are proportional to the distances from $ s $ in $ G' $ (independent of any hop bound), i.e., $ \dist_{G'}^h (s, v) = \dist_{G'} (s, v) $.
By combining these two facts (see Lemma~\ref{lem:combining to satisfy conditions}) it follows that the distance estimates returned by the auxiliary algorithm satisfy conditions~\eqref{eq:reduction approximation} and~\eqref{eq:reduction domination} of Theorem~\ref{thm:approximate to exact reduction distributed}.

\begin{lemma}\label{lem:approximation first auxiliary algorithm distributed}
For every node $ v \in V $, $ \dist_G (s, v) \leq \dist_{G'} (s, v) \leq 2 \cdot \dist_G (s, v) $.
\end{lemma}

\begin{proof}
First, observe that by the definition of $ G' $ we clearly have $ \dist_{G'} (s, v) \leq 2 \cdot \dist_G (s, v) $ for every node $ v \in V $ as every path of $ G $ is also contained in $ G' $ and the weight of such a path in $ G' $ is at most twice its weight in $ G $ by the definition of the edge weights in $ G' $.

We now show that $ \dist_{G'} (s, v) \geq \dist_G (s, v) $ for every node $ v \in V $.
Observe that it is sufficient to show that $ w_{G'} (u, v) \geq \dist_G (u, v) $ for every edge $ (u, v) $ of $ G' $ as then every path from $ s $ to $ v $ in $ G' $ has weight at least $ \dist_G (s, v) $, which in particular also applies to the shortest path from $ s $ to $ v $ in $ G' $.
Note that by the definition of the edge weights in $ G' $ we now only have to show that $ \dist_H (s, x) \geq \dist_G (s, x) $ for every skeleton node $ x \in C $.
Now observe that for every pair of skeleton nodes $ x, y \in C $ we have $ \tilde \dist (x, y) \geq \dist_G (x, y) $ by~\eqref{eq:upper and lower bound weighted BFS} and thus $ w_H (x, y) \geq \dist_G (x, y) $.
This in turn implies that $ \dist_H (s, x) \geq \dist_G (s, x) $ for every skeleton node $ x \in C $, as desired.
\end{proof}

\begin{lemma}\label{lem:triangle inequality first auxiliary algorithm distributed}
For every node $ v \in V $, $ \dist_{G'}^h (s, v) = \dist_{G'} (s, v) $.
\end{lemma}

\begin{proof}
The inequality $ \dist_{G'} (s, v) \leq \dist_{G'}^{h} (s, v) $ is obvious by the definition of the $h$-hop distance from $ s $ to $ v $.
In the remainder of this proof we argue that $ \dist_{G'}^{h} (s, v) \leq \dist_{G'} (s, v) $ by constructing a path in~$ G'$ with at most $ h $~edges and of weight at most $ \dist_{G'} (s, v) $.

Consider a shortest path $ \pi' $ from $ s $ to $ v $ in $ G' $.
We can assume that $ \pi' $ is a simple path and thus $ \pi' $ contains at most one edge $ (s, x) \in (\{ s \} \times C) \setminus E $ for any $ x \in C $, and if so, this edge must be the first edge of $ \pi' $.
We assume in the following that $ \pi' $ starts with some edge $ (s, x) \in (\{ s \} \times C) \setminus E $, as for the case that all edges of $ \pi' $ are contained in $ E $ a simpler version of the argument applies.

Now let $ \pi $ be the shortest path from $ x $ to $ v $ in $ G $.
Subdivide $ \pi $ into consecutive subpaths $ \pi_1, \ldots, \pi_k $ such that the subpaths $ \pi_2, \ldots, \pi_k $ consist of exactly $ \lceil h/2 \rceil $ nodes and the subpath $ \pi_1 $ consists of at most $ \lceil h/2 \rceil $ nodes.
By the properties of $ C $ guaranteed in \ref{step:construct hitting set}, we can assume that each of the subpaths $ \pi_2, \ldots, \pi_k $ contains a skeleton node of $ C $.
Set $ y_1 = x $ and for every $ 2 \leq i \leq k $ let $ y_i $ be a skeleton node on $ \pi_i $.
Note that between any pair of consecutive skeleton nodes $ y_i $ and $ y_{i+1} $ (for $ 1 \leq i \leq k-1 $) there are at most $ h-1 $ edges on $ \pi $ and thus
\begin{equation}
\dist_G^{h} (y_i, y_{i+1}) = \dist_G (y_i, y_{i+1}) \, . \label{eq:centers on subpaths}
\end{equation}

We now give an upper bound on the weight of the edge $ (s, y_k) $ in $ G' $, mainly applying the triangle inequality for the distance metric induced by the skeleton graph $ H $:
\begin{align*}
w_{G'} (s, y_k) &\leq \dist_H (s, y_k)				&& \text{(definition of $ w_{G'} (s, y_k) $)} \\
 &\leq \dist_H (s, y_1) + \dist_H (y_1, y_2) + \dots + \dist_H (y_{k-1}, y_k)				&& \text{(triangle inequality)} \\
 &\leq w_H (s, y_1) + w_H (y_1, y_2) + \dots + w_H (y_{k-1}, y_k)				&& \text{($ \dist_H (y_i, y_{i+1}) \leq w_H (y_i, y_{i+1}) $)} \\
 &= w_H (s, y_1) + \tilde{\dist} (y_1, y_2) + \dots + \tilde{\dist} (y_{k-1}, y_k)				&& \text{(definition of $ w_H (y_i, y_{i+1}) $)} \\
 &\leq w_H (s, y_1) + 2 \cdot \dist_G^h (y_1, y_2) + \dots + 2 \cdot \dist_G^h (y_{k-1}, y_k)				&& \text{(by \eqref{eq:upper and lower bound weighted BFS})} \\
 &= w_H (s, y_1) + 2 \cdot \dist_G (y_1, y_2) + \dots + 2 \cdot \dist_G (y_{k-1}, y_k)				&& \text{(by \eqref{eq:centers on subpaths})} \\
 &= w_H (s, y_1) + 2 \cdot \dist_G (y_1, y_k)				&& \text{($ y_1, \dots, y_k $ on shortest path $ \pi $)} \\
 &= w_H (s, x) + 2 \cdot \dist_G (x, y_k)				&& \text{($ y_1 = x $)}
\end{align*}

Now consider the path $ \pi'' $ in $ G' $ consisting of first the edge from $ s $ to $ y_k $ of weight $ w_{G'} (s, y_k) = \dist_H (s, y_k) $ and then the subpath of $ \pi $ from $ y_k $ to $ v $.
We now compare the weight of $ \pi' $ with the weight of $ \pi'' $.
Recall that $ \pi' $ consists of first the edge $ (s, x) $ and then a path from $ x $ to $ v $ consisting only of edges contained in $ E $.
Therefore, $ \pi' $ has weight at least
\begin{equation*}
w_H (s, x) + 2 \cdot \dist_G (x,v) \, .
\end{equation*}
By the upper bound on $ \dist_H (s, y_k) $ above, $ \pi'' $ has weight at most
\begin{equation*}
w_H (s, x) + 2 \cdot \dist_G (x, y_k) + 2 \cdot \dist_G (y_k, v) = w_H (s, x) + 2 \cdot \dist_G (x, v) \, .
\end{equation*}
It follows that the weight of $ \pi'' $ is at most the weight of $ \pi' $, where $ \pi' $ was the shortest path from $ s $ to $ v $ in~$ G' $.
Furthermore, $ \pi'' $ consist of at most $ h $ edges.
Thus, $ \dist_{G'}^{h} (s, v) \leq \dist_{G'} (s, v) $ as desired.
\end{proof}

\begin{lemma}\label{lem:combining to satisfy conditions}
For every node $ v \in V $, the distance estimate $ \hat{\dist} (s, v) $ satisfies conditions~\eqref{eq:reduction approximation} and~\eqref{eq:reduction domination} of Theorem~\ref{thm:approximate to exact reduction distributed}.
\end{lemma}

\begin{proof}
By Lemma~\ref{lem:approximation first auxiliary algorithm distributed} we have $ \tfrac{1}{2} \cdot \dist_G (s, v) \leq \tfrac{1}{2} \cdot \dist_{G'} (s, v) \leq \dist_G (s, v) $ for every node $ v \in V $.
Now observe that by Lemma~\ref{lem:triangle inequality first auxiliary algorithm distributed} $ \tfrac{1}{2} \cdot \dist_{G'} (s, v) = \tfrac{1}{2} \cdot \dist_{G'}^h (s, v) = \hat{\dist}_G (s, v) $ and thus condition~\eqref{eq:reduction approximation} is satisfied.

Condition~\eqref{eq:reduction domination} essentially follows from the fact that the distance metric on~$ G' $ obeys the triangle inequality as together with Lemma~\ref{lem:triangle inequality first auxiliary algorithm distributed} we have, for every edge $ (u, v) \in E $:
\begin{align*}
\hat{\dist}_G (s, v) &= \tfrac{1}{2} \cdot \dist_{G'}^h (s, v) \\
 &= \tfrac{1}{2} \cdot  \dist_{G'} (s, v) \\
 &\leq \tfrac{1}{2} \cdot (\dist_{G'} (s, u) + \dist_{G'} (u, v)) \\
 &\leq \tfrac{1}{2} \cdot (\dist_{G'} (s, u) + w_{G'} (u, v)) \\
 &\leq \tfrac{1}{2} \cdot (\dist_{G'} (s, u) + 2 \cdot w_G (u, v)) \\
 &= \tfrac{1}{2} \cdot \dist_{G'} (s, u) + w_G (u, v) \\
 &= \tfrac{1}{2} \cdot \dist_{G'}^h (s, u) + w_G (u, v) \\
 &= \hat{\dist}_G (s, u) + w_G (u, v) \, . \qedhere
\end{align*}
\end{proof}

\section{First CONGEST Model Implementation}\label{sec:first distributed algorithm}

In the following we present our first exact SSSP algorithm for the CONGEST model.
Its guarantees can be formalized as follows.
\begin{theorem}\label{thm:first exact SSSP algorithm distributed}
In the Broadcast CONGEST model, there is a randomized SSSP algorithm for directed graphs with non-negative integer edge weights in the range $ \{ 0, 1, \dots, W \} $ that performs $ \tilde O (\sqrt{n D} \log W) $ rounds in expectation.
\end{theorem}

In our algorithm, we use the following implementation of the auxiliary algorithm of Section~\ref{sec:auxiliary algorithm}:
\begin{enumerate}[align=left,label=\textbf{Step \arabic*.},ref=\textbf{Step \arabic{enumi}}]
\item
We implement the sampling process of Lemma~\ref{lem:randomized hitting set}:
Simultaneously, $ s $ adds itself to $ C $ and each other node adds itself to $ C $ with probability $ \min(8 \ln{n} / h, 1) $.
Afterwards, we spend $ O (D) $ rounds in a global upcast to determine the size of $ C $.
If $ |C| > 24 n \ln{n} / h $ (a low-probability event), then abort the algorithm.
Now the set $ C $ satisfies the condition demanded in \ref{step:construct hitting set} of Section~\ref{sec:auxiliary algorithm} with high probability and thus our implementation of the auxiliary algorithm will also be correct with high probability.

\item
Using the algorithm of Lemma~\ref{lem:approximate bounded hop distances multiple sources CONGEST} as a subroutine, this step can be implemented in $ \tilde O (h \log{W} + |C|) = \tilde O (h \log{W} + n/h) $ rounds.

\item The skeleton graph is constructed implicitly -- in the sense that each skeleton node only knows its set of incoming edges together with their weight -- by globally broadcasting the set of skeleton nodes in the network in $ O (|C| + D) = \tilde O (n/h + D) $ rounds.

\item We use a message-passing version of Dijkstra's algorithm that in each iteration determines the next node to visit, i.e., the one with minimum tentative distance, by performing a global upcast in the network.
This step can thus be implemented in $ O (|C| D) = \tilde O (n D / h) $ rounds.
Note that Dijkstra's algorithm satisfies the requirement to work on graphs with non-negative edge weights.

\item By performing the corresponding edge weight modifications internally at each node, i.e., without any additional communication, the augmented graph is constructed implicitly in the sense that each node only knows its set of incoming edges together with their weight.

\item We implement this step with the Bellman-Ford algorithm.
The first iteration of the Bellman-Ford algorithm can be performed in $ O (D) $ rounds as every skeleton node $ x \in C $ already knows the weight of the edge $ (s, x) $ in $ G' $ and thus only needs to be informed about the start of the algorithm.
The remaining $ h-1 $ iterations take $ O (h) $ rounds in a synchronized implementation of Bellman-Ford by Lemma~\ref{lem:Bellman Ford CONGEST}, yielding an overall complexity of $ O (h + |C| + D) = \tilde O(h + D) $ rounds for this step.
\end{enumerate}
Asymptotically, the overall number of rounds is $ \tilde O (h \log{W} + n D / h) $.
By setting $ h = \sqrt{n D} $ we obtain an auxiliary algorithm performing $ \tilde O (\sqrt{n D} \log{W}) $ rounds.
Theorem~\ref{thm:first exact SSSP algorithm distributed} now follows as a Corollary from Lemma~\ref{lem:correctness auxiliary algorithm} and Theorem~\ref{thm:approximate to exact reduction distributed}.

\section{Second CONGEST Model Implementation}\label{sec:second distributed algorithm}

In our second distributed algorithm, we obtain better guarantees for certain parameter ranges by using a different approach for computing exact distances from $ s $ on the skeleton graph; the rest of the algorithm is the same as in Section~\ref{sec:first distributed algorithm}.
Instead of running Dijkstra's algorithm, we implement \ref{step:exact SSSP on skeleton} by repeating the algorithmic scheme and effectively constructing a skeleton of the skeleton.
Intuitively, this somewhat straightforward repetition of the scheme improves the efficiency because computing on the skeleton graph allows slightly different algorithmic techniques than computing on the original network itself as the former is simulated by performing global broadcasts in the network.
This gives us some slack to exploit for increased efficiency.
We remark that adding more levels of recursion will not boost the efficiency further because computing shortest paths for the skeleton of the skeleton is not a bottleneck in our running time analysis.

In the following, we explicitly separate the two layers -- input graph and skeleton graph -- mentioned above.
We first show how to compute SSSP ``on the skeleton graph'' and then demonstrate how this can be used for computing SSSP on the input graph.

\subsection{Implementation in Broadcast LOCAL Clique Model}

Consider the following \emph{Broadcast LOCAL Clique model} which deviates from the Broadcast CONGEST model in the following ways: (1) the communication network $ N $ is a clique, i.e., every message sent by a node is received by \emph{all} other nodes of the network and (2) the size of the message sent per round is arbitrary (and in particular may also be $ 0 $).
The complexity of an algorithm for a clique network with $ n $~nodes is determined by the number of rounds $ R (n) $ and the total size of all messages~$ M (n) $ broadcast by the nodes over the course of the algorithm.

The Broadcast LOCAL Clique model may seem a bit artificial on its own, but it is highly relevant for our CONGEST model implementation of the auxiliary algorithm because of the following straightforward simulation result for computing exact distances on the skeleton graph.

\begin{lemma}[Implicit in \cite{Nanongkai14}]\label{lem:simulation of BCC algorithm}
Assume there is an exact SSSP algorithm $ \mathcal{A} $ for directed graphs with non-negative integer edge weights in the Broadcast LOCAL Clique model spending $ R (n) $ rounds and $ M (n) $ messages.
Then \ref{step:exact SSSP on skeleton} of the auxiliary algorithm can be implemented in $ \tilde O (M (|C|) / B + R (|C|) \cdot D) $ rounds in the Broadcast CONGEST model, where $ B = \Theta (\log{n}) $ is the bandwidth of the network and $ D $ is its diameter.
\end{lemma}

\begin{proof}[Proof Sketch]
We simulate a run of the Broadcast LOCAL Clique algorithm $ \mathcal{A} $ on the skeleton graph, which is a clique of $ |C| $~nodes.
We do this by making each message sent by $ \mathcal{A} $ global knowledge, which we carry out by globally broadcasting all messages via a breadth-first-search spanning tree of the communication network $ N $.
Such a spanning tree can be constructed initially in $ O (D) $ rounds.
In the $i$-th of the $ R(|C|) $ rounds of $ \mathcal{A} $ we have to send some $ M_i $~messages and we know that $ \sum_{1 \leq i \leq R(|C|)} M_i = M (|C|) $.
In the Broadcast CONGEST model, the total number of rounds for this simulation therefore is $ \tilde O (\sum_{1 \leq i \leq R(|C|)} (M_i / B + D)) = \tilde O ( M(|C|) / B + R(|C|) \cdot D) $ by standard arguments~\cite{Peleg00}.
\end{proof}

We now show how to obtain an efficient exact SSSP algorithm in the Broadcast LOCAL Clique model by first implementing the auxiliary algorithm of Section~\ref{sec:auxiliary algorithm} and then extending it to an exact SSSP algorithm using the the recursive-scaling reduction of Theorem~\ref{thm:approximate to exact reduction distributed}.

\begin{lemma}\label{lem:first exact SSSP algorithm blackboard}
In the Broadcast LOCAL Clique model, there is a randomized SSSP algorithm for directed graphs with non-negative integer edge weights in the range $ \{ 0, 1, \dots, W \} $ that in expectation spends $ R (n) = \tilde O (h \log W) $ rounds and $ M (n) = \tilde O ((n h + n^2 / h) \log{W}) $ messages for any integer parameter $ 1 \leq h \leq n $.
\end{lemma}

\begin{proof}
We implement \ref{step:construct hitting set}, \ref{step:construct skeleton graph}, and \ref{step:construct augmented graph} as in Section~\ref{sec:first distributed algorithm}.
We implement \ref{step:compute approximate hop distances} by simultaneously running an instance $ \mathcal{A}_x $ of the algorithm of Lemma~\ref{lem:approximate bounded hop distances CONGEST} from each skeleton node $ x \in C $.
In each of the $ \tilde O (h \log{W}) $ rounds of the algorithm of Lemma~\ref{lem:approximate bounded hop distances CONGEST}, we aggregate, for every node, all the messages that it would have to broadcast over all instances $ \mathcal{A}_x $ and broadcast them to all other nodes in a single round in the Broadcast LOCAL Clique model. 
This results in a total size of $ \tilde O (n |C| \log{W}) = \tilde O (n^2 \log{W} / h) $ for all these messages in the Broadcast LOCAL Clique model.
We implement \ref{step:exact SSSP on skeleton} in the naive way by having each center node broadcast its outgoing edges in $ H $ and computing $ \dist_H (s, \cdot) $ internally at every node.
This takes $ O(1) $ rounds and $ O (|C|^2) = \tilde O (n^2/h^2) $ messages.
Finally, \ref{step:final Bellman-Ford} requires a single run of Bellman-Ford, which by Lemma~\ref{lem:Bellman Ford CONGEST} takes $ O (h) $ rounds and $ O (n h) $ messages.
We now apply the reduction of Theorem~\ref{thm:approximate to exact reduction distributed} to obtain the desired SSSP algorithm for the Broadcast LOCAL Clique model.
\end{proof}

We remark that the essential bottleneck of this implementation of the auxiliary turns out to be the final Bellman-Ford computation in \ref{step:final Bellman-Ford} when performed in this algorithm for the Broadcast LOCAL Clique model.
Recall that this step is necessary to ensure the domination property of~\eqref{eq:reduction domination}.

\subsection{Faster CONGEST Model Implementation in High-Diameter Networks}

\begin{theorem}\label{thm:second exact SSSP algorithm CONGEST}
In the Broadcast CONGEST model, there is a randomized SSSP algorithm for directed graphs with non-negative integer edge weights in the range $ \{ 0, 1, \dots, W \} $ that performs $ \tilde O ((\sqrt{n} D^{1/4} + n^{3/5} + D) \log{W}) $ rounds in expectation.
\end{theorem}

\begin{proof}
We implement \ref{step:construct hitting set}, \ref{step:compute approximate hop distances}, \ref{step:construct skeleton graph}, \ref{step:construct augmented graph}, and \ref{step:final Bellman-Ford} of the auxiliary algorithm as in Section~\ref{sec:first distributed algorithm}.
Applying the simulation of Lemma~\ref{lem:simulation of BCC algorithm} to the algorithm of Lemma~\ref{lem:first exact SSSP algorithm blackboard} (with parameter~$ h' $), we can implement \ref{step:exact SSSP on skeleton} in $ \tilde O (|C| h' + |C|^2 / h' + h' D) $ rounds.
Now the overall number of rounds is $ \tilde O (h \log{W} + |C| + |C| h' + |C|^2 / h' + h' D) = \tilde O (h \log{W} + n / h + n h' / h + n^2 / (h^2 h') + h' D) $.
We use two variants for balancing these terms.
In the first variant, we set $ h = \sqrt{n} D^{1/4} $ and $ h' = \max (1, \sqrt{n} / D^{3/4}) $ to get an upper bound of $ \tilde O (\sqrt{n} D^{1/4} + \max (\sqrt{n}/D^{1/4}, n/D) + D) $ rounds, where the first term dominates the second term when $ D \geq n^{2/5} $.
In the second variant, we set $ h = n^{3/5} $ and $ h' = n^{1/5} $ to get an upper bound of $ \tilde O (n^{3/5} + n^{1/5} D) $ rounds, where the first term dominates the second term when $ D \leq n^{2/5} $.
Our overall algorithm first computes a $2$-approximation to the diameter $ D $ of the communication network $ N $ in $ O (D) $ rounds by performing breadth-first search from an arbitrary node in $ N $ and then chooses $ h $ and $ h' $ according to the approximate value of $ D $ to get an upper bound of $ \tilde O (\sqrt{n} D^{1/4} + n^{3/5} + D) $ rounds in this implementation of the auxiliary algorithm.
We now apply the reduction of Theorem~\ref{thm:approximate to exact reduction distributed} to obtain the desired SSSP algorithm for the Broadcast CONGEST model.
\end{proof}

Note that the algorithm of Theorem~\ref{thm:second exact SSSP algorithm CONGEST} is (asymptotically) faster than the simpler algorithm of Theorem~\ref{thm:first exact SSSP algorithm distributed} when $ D = \omega (n^{1/5}) $.

\section{Additional Results}\label{sec:additional_results}

In this section, we work out some additional results, first for approximate SSSP the distributed setting, and then for exact SSSP in the parallel setting.

\subsection{Directed Approximate SSSP}

In the following, we give an algorithm for computing approximate SSSP on directed graphs in the CONGEST model that matches the round complexity of the fastest known algorithm for single-source reachability up to polylogarithmic factors.

The algorithmic scheme followed by our algorithm is quite similar to the one of the auxiliary algorithm:
\begin{enumerate}[align=left,label=\textbf{Step \arabic*.},ref=\textbf{Step \arabic*}]
\item \label{step:construct hitting set approx}
Construct a set $ C \subseteq V $ of \emph{skeleton nodes} containing (a) the source node $ s $ and (b) additionally, for every pair of nodes $ u $ and $ v $ such that the shortest path from $ u $ to $ v $ in $ G $ consists of exactly $ \lceil h/2 \rceil $ nodes, at least one node on one of these shortest paths.

\item \label{step:compute approximate hop distances approx}
For each skeleton node $ x \in C $, compute $ (1 + \epsilon/3) $-approximate $h$-hop distances from $ x $, i.e., distance estimates $ \tilde{\dist} (x, \cdot) $ such that
\begin{equation*}
\dist_G (x, v) \leq \tilde{\dist} (x, v) \leq (1 + \epsilon) \cdot \dist_G^h (x, v)
\end{equation*}
for every node $ v \in V $.

\item \label{step:construct skeleton graph approx}
Construct the \emph{skeleton graph} $ H = (C, C^2, w_H) $ with edge weight $ w_H (x, y) = \tilde{\dist} (x, y) $ for every $ (x, y) \in C^2 $.

\item \label{step:approx SSSP on skeleton}
Compute $ (1 + \epsilon/3) $-approximate distances from $ s $ on the skeleton graph~$ H $, i.e., $ \dist_H (s, x) $ 
\begin{equation*}
\dist_H (s, x) \leq \dist' (s, x) \leq (1 + \epsilon) \cdot \dist_H (s, x)
\end{equation*}
for every skeleton node $ x \in C $.

\item \label{step:final approximation}
Compute
\begin{equation*}
\hat{\dist} (s, v) := \min_{x \in C} (\dist' (s, x) + \tilde{\dist} (x, v))
\end{equation*}
for every node $ v \in V $.
\end{enumerate}

\begin{lemma}
For any directed input graph $ G = (V, E, w_G) $ with fixed source node $ s $, any $ 0 < \epsilon \leq 1 $ and any integer $ 1 \leq h \leq n $, the algorithm above consisting of \textbf{Steps 1--5} computes, for every node $ v \in V $, a distance estimate $ \hat{\dist} (s, v) $ such that $ \dist_G (s, v) \leq \hat{\dist} (s, v) \leq (1 + \epsilon) \dist_G (s, v) $.
\end{lemma}

We omit the proof of this lemma, as the arguments in correctness proof for the algorithm are just a variation of those given in~\cite{Nanongkai14}.
The $ (1 + \epsilon) $-approximation guarantee follows because $ (1 + \epsilon/3)^2 \leq 1 + \epsilon $.

We proceed with giving a CONGEST-model implementation of the algorithm above.
Here, we use a two-step process similar to the implementation of the auxiliary algorithm in Section~\ref{sec:second distributed algorithm}:
We first provide an implementation in the Broadcast LOCAL Clique model and then use that algorithm as a black box for the implementation in the Broadcast CONGEST model.

\begin{lemma}\label{lem:approx SSSP algorithm blackboard}
In the Broadcast LOCAL Clique model, there is a randomized $ (1 + \epsilon) $-approximate SSSP algorithm for directed graphs with non-negative integer edge weights in the range $ \{ 0, 1, \dots, W \} $ that spends $ R (n) = \tilde O (h \log W / \epsilon) $ rounds and $ M (n) = \tilde O (n^2 / (\epsilon h) \log{W}) $ messages for any integer parameter $ 1 \leq h \leq n $.
The algorithm is correct with high probability.
\end{lemma}

\begin{proof}
We implement \ref{step:construct hitting set approx} and \ref{step:construct skeleton graph approx} as in Section~\ref{sec:first distributed algorithm}.
We implement \ref{step:compute approximate hop distances approx} by simultaneously running an instance $ \mathcal{A}_x $ of the algorithm of Lemma~\ref{lem:approximate bounded hop distances CONGEST} from each skeleton node $ x \in C $.
In each of the $ \tilde O (h \log{W} / \epsilon) $ rounds of the algorithm of Lemma~\ref{lem:approximate bounded hop distances CONGEST}, we aggregate, for every node, all the messages that it would have to broadcast over all instances $ \mathcal{A}_x $ and broadcast them to all other nodes in a single round in the Broadcast LOCAL Clique model. 
This results in a total size of $ \tilde O (n |C| \log{W} / \epsilon) = \tilde O (n^2 \log{W} / (\epsilon h)) $ for all these messages in the Broadcast LOCAL Clique model.
We implement \ref{step:approx SSSP on skeleton} in the naive way by having each center node broadcast its outgoing edges in $ H $ and computing $ \dist' (s, \cdot) := \dist_H (s, \cdot) $ (i.e., exact distances) internally at every node.
This takes $ O(1) $ rounds and $ O (|C|^2) = \tilde O (n^2/h^2) $ messages.
Finally, \ref{step:final approximation} only requires internal computation.
\end{proof}

\begin{theorem}
In the Broadcast CONGEST model, there is a randomized $ (1 + \epsilon) $-approximate SSSP algorithm for directed graphs with non-negative integer edge weights in the range $ \{ 0, 1, \dots, W \} $ that performs $ \tilde O ((\sqrt{n} D^{1/4} + D) \log{W} / \epsilon) $ rounds.
The algorithm is correct with high probability.
\end{theorem}

\begin{proof}
We implement \ref{step:construct hitting set approx}, \ref{step:compute approximate hop distances approx}, and \ref{step:construct skeleton graph approx} as in Section~\ref{sec:first distributed algorithm}.
Applying the simulation of Lemma~\ref{lem:simulation of BCC algorithm}, which holds regardless of the approximation ratio, to the algorithm of Lemma~\ref{lem:approx SSSP algorithm blackboard} (with parameters~$ h' $ and $ \epsilon' = \epsilon/3 $), we can implement \ref{step:approx SSSP on skeleton} in $ \tilde O (|C|^2 / (\epsilon h') \log W + h' D \log W / \epsilon) $ rounds.
Now the overall number of rounds is $ \tilde O ((h + |C| + |C|^2 / h' + h' D) \log{W} / \epsilon) = \tilde O ((h + n / h + n^2 / (h^2 h') + h' D) \log W / \epsilon) $.
By setting $ h = \sqrt{n} D^{1/4} $ and $ h' = \max(1, \sqrt{n} / D^{3/4}) $ this becomes $ \tilde O ((\sqrt{n} D^{1/4} + D) \log{W} / \epsilon) $ as desired.
\end{proof}

\subsection{PRAM Trade-off}

In the following, we show that our approach also gives a new work/depth trade-off for computing SSSP in the PRAM model.
In analogy to the RAM model, the PRAM (parallel RAM) allows the possibility of processors computing in parallel with shared memory access.
The \emph{work} of a parallel algorithm is the total number of operations performed over all processors
and the \emph{depth} is the length of the longest series of operations that have to be performed sequentially due to data dependencies.
The depth is essentially the number of parallel computation steps needed until the last processor is finished for any schedule of assigning the algorithm's operations to the processors.

Our contribution in the PRAM model is an exact SSSP algorithm with work $ \tilde O ((n^3/h^3 + m n/h) \log W) $ and depth $ \tilde O (h) $ for any $ 1 \leq h \leq \sqrt{n} $.
For \emph{directed} graphs, this specific trade-off was not known before.
In particular, the algorithm of Klein and Subramanian~\cite{KleinS97}, who follows the same general framework as we do, has work $ \tilde O (m \sqrt{n}) $ and depth $ \tilde O (\sqrt{n}) $, i.e., it it does not allow for such a trade-off as any choice of $ h $ different from $ \tilde \Theta (\sqrt{n}) $ would be sub-optimal.
Thus, the conceptual novelties in our algorithm, which were motivated by the application to the CONGEST model, also carry over to an improvement in the PRAM model.

\begin{theorem}
In the PRAM model, there is a randomized SSSP algorithm for directed graphs with non-negative integer edge weights in the range $ \{ 0, 1, \dots, W \} $ that has work $ \tilde O ((n^3/h^3 + m h + m n/h) \log W) $ with high probability and in expectation and depth $ \tilde O (h) $ for any given $ 1 \leq h \leq n $.
\end{theorem}

\begin{proof}
To obtain the exact SSSP algorithm, we implement the auxiliary algorithm of Section~\ref{sec:auxiliary algorithm} as follows:
\begin{enumerate}[align=left,label=\textbf{Step \arabic*.},ref=\textbf{Step \arabic{enumi}}]
\item We implement the sampling process of Lemma~\ref{lem:randomized hitting set}:
Each node is added to $ C $ independently with probability $ \min (8 \ln n/h, 1) $.
Afterwards, determine the size of $ C $ and if $ |C| > 24 n \ln n/h $ (a low-probability event), then abort the algorithm.
This has work $ O (n) $ and depth $ O (\log n) $.
Now the set $ C $ satisfies the condition demanded in \ref{step:construct hitting set} of Section~\ref{sec:auxiliary algorithm} with high probability and thus our implementation of the auxiliary algorithm will also be correct with high probability.

\item For each node in $ C $, this step can be implemented by computing a shortest path tree up to distance $ O (h) $ in a graph with suitably rounded edge weights.
Similar to the analysis of bounded breadth-first search, this approach has work $ \tilde O (m \log W) $ for each node in $ C $ and depth $ O (h) $ (see Lemma 3.2 in~\cite{KleinS97}, which is similar to Lemma~\ref{lem:approximate bounded hop distances CONGEST} in Section~\ref{sec:preliminaries}).
Thus the total work of this step is $ \tilde O ((n/h) \cdot m \log W) $ and the depth is $ O (h) $.

\item The straightforward approach for constructing and storing the graph $ H $ has work $ O (|C|^3) = O (n^3/h^3) $ and depth $ O (\log n) $.

\item We implement this step by using the all-pairs shortest path algorithm based on repeated squaring with min-plus matrix multiplication.
This has work $ \tilde O (|C|^3) = \tilde O (n^3 / h^3) $ and depth $ \polylog{n} $.

\item The straightforward approach for constructing and storing the graph $ G' $ has work $ O (m + |C|) = O (m) $ and depth $ O (\log n) $.

\item We implement this step with a parallel version of the Bellman-Ford algorithm that synchronizes the computation in $ h $ rounds.
This has work $ \tilde O ((m + |C|) \cdot h) = \tilde O (m h) $ and depth $ \tilde O (h) $.
\end{enumerate}

Overall, our auxiliary algorithm has work $ \tilde O (n^3/h^3 + m h + m (n/h) \log W) $ and depth $ \tilde O (h) $.
By applying the reduction of Klein and Subramanian (see Lemma 4.1 in~\cite{KleinS97}, which is similar to Lemma~\ref{thm:approximate to exact reduction distributed} in Section~\ref{sec:preliminaries}), we obtain an exact SSSP algorithm that has work $ \tilde O ((n^3/h^3 + m h + m n/h) \log W) $ and depth $ \tilde O (h) $.
As the auxiliary algorithm is correct with high probability, the work bound of the exact SSSP algorithm holds with high probability (and in expectation).
\end{proof}

\section{Conclusion and Discussion}

In this paper, we have given two new algorithms for computing SSSP in the CONGEST model, the first performing $ \tilde O (\sqrt{n D}) $ rounds and the second performing $ \tilde O (\sqrt{n} D^{1/4} + n^{3/5} + D) $ rounds.
Our algorithms follow a general scaling approach for reducing the problem to approximate distance computation.
However, not any arbitrary approximate SSSP algorithm is adequate for this reduction as the returned distance estimates must additionally `dominate' the edge weights of the graph.
Our main contribution is a schematic auxiliary algorithm that computes suitable distance estimates, which in particular extends and improves upon prior work in the PRAM model.
We then provide different ways of implementing this scheme.

Currently, a natural barrier for directed SSSP algorithms is $ \tilde O (\sqrt{n} D^{1/4} + D) $ rounds as this is the running time of the fastest known algorithm for single-source reachability~\cite{GhaffariU15}.
We do meet this barrier in our approximation algorithm, but when solving SSSP exactly, \ref{step:final Bellman-Ford} of our auxiliary algorithm is a bottleneck: we cannot implement it fast enough on the skeleton graph to get the $ \tilde O (\sqrt{n} D^{1/4} + D) $ bound.
A simple counter-example (see Fig.~\ref{fig:approximation counterexample}) shows that \ref{step:final Bellman-Ford} cannot simply be omitted.
Finding an SSSP algorithm performing only $ \tilde O (\sqrt{n} D^{1/4} + D) $ rounds therefore remains a major open problem.
In the same respect, it is also very interesting to understand whether $ \tilde \Theta (n^{1/2}D^{1/4} + D) $ is the right bound:
it seems both challenging to break the $ n^{1/2} D^{1/4} $ bound for the reachability problem or to prove an $ n^{1/2}D^c $ lower bound, for any constant~$ c $, for the directed exact SSSP problem. 

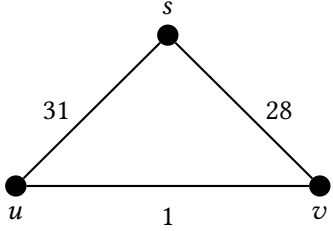
\begin{figure}
\centering
\begin{tikzpicture}
\tikzstyle{vertex}=[circle,fill=black,minimum size=8pt,inner sep=0pt,outer sep=0pt]

\node [vertex,label=above:$s$] (s) at (0, 0) {};
\node [vertex,label=below:$u$] (u) at (-2, -2) {};
\node [vertex,label=below:$v$] (v) at (2, -2) {};

\draw [thick] (s) edge node [label=left:$31$] {} (u);
\draw [thick] (s) edge node [label=right:$28$] {} (v);
\draw [thick] (u) edge node [label=below:$1$] {} (v);
\end{tikzpicture}
\caption{This graph $ G $ shows that \ref{step:final Bellman-Ford} of our auxiliary algorithm cannot be omitted. Using the algorithm of Lemma~\ref{lem:approximate bounded hop distances CONGEST} with $ \epsilon = 1 $ and $ h = 2 $ we obtain distance estimates $ \tilde d (s, u) = 28 $ and $ \tilde d (s, v) = 32 $.
However, $ \tfrac{1}{2} \cdot \tilde d (s, v) > \tfrac{1}{2}  \cdot \tilde d (s, u) + w_G (u, v) $ and thus these distance estimates do not dominate the edge weights as demaned in the preconditions of Theorem~\ref{thm:approximate to exact reduction distributed}.
}\label{fig:approximation counterexample}
\end{figure}

Another challenge lies in finding a \emph{deterministic} SSSP algorithm that outperforms the distributed implementation of the Bellman-Ford algorithm with its round complexity of $ O (n) $.
In our specific scaling approach, the deterministic procedure of Henzinger, Krinninger, and Nanongkai~\cite{HenzingerKN16} for computing a skeleton graph cannot be applied.
This technique only works for undirected graphs and in our scaling algorithm we necessarily have to deal with directed graphs, even if the input graph is undirected.
The recent deterministic construction in the APSP algorithm of Agarwal et al.~\cite{AgarwalRKP18} is also not an option because it takes $ \Omega (n) $ rounds.

In general, it is very appealing to understand how fast we can solve graph problems exactly in the CONGEST model.
In the past few years there have been many tight approximation algorithms developed in this model, but not much was known for exact algorithms.
Two particular problems, as mentioned in~\cite{HuangNS17}, are exact maximum matching and minimum cut; the former was recently shown to admit an algorithm performing $ O (n \log n) $ rounds~\cite{AhmadiKO18} while no non-trivial algorithm is known for the latter.

\section*{Acknowledgment}

This project has received funding from the European Research Council (ERC) under the European
Unions Horizon 2020 research and innovation programme under grant agreement No 715672. Danupon
Nanongkai was also partially supported by the Swedish Research Council (Reg.\ No.\ 2015-04659).
We thank the anonymous reviewers of FOCS 2018 for their helpful comments.
We thank Cliff Liu for making us aware of a formal problem in our treatment of $0$-weight edges in a previous version of this paper.

\appendix
\section*{Appendix}
\section{From Approximate to Exact Solutions}\label{apx:reduction from approximate to exact}

In the following, we give a proof of Theorem~\ref{thm:approximate to exact reduction distributed}, i.e., we show how the auxiliary algorithm for directed graphs with non-negative integer edge weights can be extended to computing exact SSSP.
The reduction is given in the paper by Klein and Subramanian~\cite{KleinS97}, but some errors introduced in the typesetting might make the corresponding part of their journal paper hard to read; we therefore include the proof for the readers' convenience.

\begin{proof}[Proof of Theorem~\ref{thm:approximate to exact reduction distributed}]
We design a recursive algorithm $ \mathcal{B} $ that, given a  weighted directed graph $ G = (V, E, w_G) $ with source nodes $ s $ and an upper bound $ \Delta \geq 0 $, computes SSSP in a weighted directed graph if the distance from the source node $ s $ to any other node is at most $ \Delta $, where initially $ \Delta = nW $.

If $ \Delta < 1 $, then the algorithm returns $ \tilde{\dist} (s, v) = 0 $ for every node $ v \in V $.
Otherwise, the algorithm proceeds as follows:
\begin{enumerate}
\item Set $ \delta := \max(\lfloor \log{\Delta} - \log{n} - 1 \rfloor, 0) $.
\item Construct $ \hat{G} = (V, \hat{E}, w_{\hat{G}}) $ by setting $ \hat{E} := \{ e \in E : w_G (e) \leq \Delta \} $ and
\begin{equation*}
w_{\hat{G}} (e) := \left\lfloor \tfrac{w_G (e)}{2^\delta} \right\rfloor
\end{equation*}
for every edge $ e \in \hat{E} $.
\item Compute a distance estimate $ \hat{\dist} (s, v) $ for every node $ v \in V $ by running the auxiliary algorithm $ \mathcal{A} $ on $ \hat{G} $.
\item Construct $ G' = (V, E, w_{G'}) $ by setting
\begin{equation*}
w_{G'} (u, v) := w_G (u, v) + 2^\delta \hat{\dist} (s, u) - 2^\delta \hat{\dist} (s, v)
\end{equation*}
for every edge $ (u, v) \in E $.
\item Recursively compute $ \dist_{G'} (s, v) $ for every node $ v \in V $ by running $ \mathcal{B} $ on $ G' $ with $ \Delta' := \lfloor \tfrac{3}{4} \Delta \rfloor $.
\item Return $ \tilde{\dist} (s, v) := \dist_{G'} (s, v) + 2^\delta \hat{\dist} (s, v) $ for every node $ v \in V $.
\end{enumerate}

Observe that $ \Delta $ reduces by a $ \tfrac{1}{4} $-fraction in each recursive call and thus the number of repetitions is $ O (\log{n W}) $.
Furthermore, the weight of each edge $ e \in \hat{E} $ consists of $ O (\log{n}) $ bits and thus the algorithm~$ \mathcal{A'} $ is called on a graph with integer edge weights in the range $ \{0, 1, \dots, O (n) \} $.
Obviously, all operations of~$ \mathcal{B} $, except for the calls~to $ \mathcal{A'} $ can be implemented in $ \tilde O (D \log{W}) $ broadcast rounds in the CONGEST model.
To establish the correctness of the algorithm, it is left to prove that the returned distances are correct, that the edge weights in the graph $ G' $ are non-negative, and that the distances from the source in $ G' $ are indeed bounded by $ \Delta' $.

First, we show that $ \tilde{\dist} (s, v) = \dist_G (s, v) $ for every node $ v \in V $.
Consider a path $ \pi $ from $ s $~to~$ v $ in~$ G $.
\begin{align*}
w_{G'} (\pi) &= \sum_{e \in \pi} w_{G'} (e) \\
 &= \sum_{e \in \pi} \left[ w_G (e) + 2^\delta \hat{\dist} (s, u) - 2^\delta \hat{\dist} (s, v) \right] \\
 &= 2^\delta \hat{\dist} (s, s) - 2^\delta \hat{\dist} (s, v) + \sum_{e \in \pi} w_G (e) \\
 &= w_G (\pi) - 2^\delta \hat{\dist} (s, v)
\end{align*}
Now let $ \pi^* $ be the shortest path from $ s $ to $ v $ in $ G $.
By the arguments above we have $ w_{G'} (\pi^*) = w_G (\pi^*) - 2^\delta \hat{\dist} (s, v) = \dist_G (s, v) - 2^\delta \hat{\dist} (s, v) $ and for every other path $ \pi $ from $ s $ to $ v $ we have $ w_{G'} (\pi) = w_G (\pi) - 2^\delta \hat{\dist} (s, v) \geq w_G (\pi^*) - 2^\delta \hat{\dist} (s, v) $.
Thus, $ \pi^* $ is also a shortest path from $ s $ to $ v $ in $ G' $, i.e. $ \dist_{G'} (s, v) = w_{G'} (\pi^*) $.
We therefore have
\begin{equation}
\dist_G (s, v) = \dist_{G'} (s, v) + 2^\delta \hat{\dist} (s, v) \label{eq:correctness estimate additive}
\end{equation}
as desired.

Second, we show that the weights of the graph on which the algorithm is called recursively are non-negative, i.e., $ w_{G'} (u, v) \geq 0 $ for every edge $ (u, v) \in E' $.
By~\eqref{eq:reduction domination} we have
\begin{equation*}
\hat{\dist} (s, v) \leq \hat{\dist} (s, u) + w_{\hat{G}} (u, v)
\end{equation*}
which is equivalent to
\begin{equation*}
2^\delta \hat{\dist} (s, v) \leq 2^\delta \hat{\dist} (s, u) + 2^\delta w_{\hat{G}} (u, v) \, .
\end{equation*}
We now have
\begin{equation*}
w_{G'} (u, v) = w_G (u, v) + 2^\delta \hat{\dist} (s, u) - 2^\delta \hat{\dist} (s, v) \geq w_G (u, v) - 2^\delta w_{\hat{G}} (u, v) \geq w_G (u, v) - w_G (u, v) = 0
\end{equation*}
as desired.

Third, we show that $ \dist_{G'} (s, v) \leq \Delta' $ for every node $ v \in V $.
Let $ \pi $ be a shortest path from $ s $~to~$ v $ in~$ \hat{G} $.
By \eqref{eq:reduction approximation}, and the definition of the edge weights in $ \hat{G} $, we have
\begin{equation*}
\hat{\dist} (s, v) \geq \tfrac{1}{2} \cdot \dist_{\hat{G}} (s, v) = \tfrac{1}{2} \cdot w_{\hat{G}} (\pi) = \tfrac{1}{2} \cdot \sum_{e \in \pi} w_{\hat{G}} (e) = \tfrac{1}{2} \cdot \sum_{e \in \pi} \left\lfloor \tfrac{w_G (e)}{2^\delta} \right\rfloor \, .
\end{equation*}
If $ \delta > 0 $, $ \delta = \lfloor \log{\Delta} - \log{n} - 1 \rfloor $ and $ 2^\delta \leq 2^{\log{(\Delta / (2n))}} = \tfrac{\Delta}{2 n} $.
We thus have
\begin{align*}
2^\delta \cdot \sum_{e \in \pi} \left\lfloor \tfrac{w_G (e)}{2^\delta} \right\rfloor &\geq 2^\delta \cdot \sum_{e \in \pi} (\tfrac{w_G (e)}{2^\delta} - 1) \\
 &\geq \sum_{e \in \pi} w_G (e) - 2^\delta n \\
 &= w_G (\pi) - 2^\delta n \\
 &\geq w_G (\pi) - \tfrac{1}{2} \Delta
\end{align*}
If $ \delta = 0 $, then $ 2^\delta \cdot \sum_{e \in \pi} \left\lfloor \tfrac{w_G (e)}{2^\delta} \right\rfloor = \sum_{e \in \pi} w_G (e) = w_G (\pi) $.
In any case we have
\begin{equation*}
2^\delta \cdot \sum_{e \in \pi} \left\lfloor \tfrac{w_G (e)}{2^\delta} \right\rfloor \geq w_G (\pi) - \tfrac{1}{2} \Delta \, .
\end{equation*}
Thus, by \eqref{eq:correctness estimate additive}, we get
\begin{align*}
\dist_{G'} (s, v) &= \dist_G (s, v) - 2^\delta \hat{\dist} (s, v) \\
 &\leq \dist_G (s, v) - \tfrac{1}{2} \cdot 2^\delta \cdot \sum_{e \in \pi} \left\lfloor \tfrac{w_G (e)}{2^\delta} \right\rfloor \\
 &\leq \dist_G (s, v) - \tfrac{1}{2} \cdot w_G (\pi) + \tfrac{1}{4} \Delta \\
 &\leq \dist_G (s, v) - \tfrac{1}{2} \cdot \dist_G (s, v) + \tfrac{1}{4} \Delta \\
 &= \tfrac{1}{2} \cdot \dist_G (s, v) + \tfrac{1}{4} \Delta \\
 &\leq \tfrac{1}{2} \Delta + \tfrac{1}{4} \Delta \\
  &= \tfrac{3}{4} \Delta = \Delta'
\end{align*}
as desired.

Finally, observe that the claimed Las Vegas guarantee can be achieved by adding the following verification step at the end of the algorithm:
(1) every node $ v $ sends its result $ \tilde{\dist} (s, v) $ to its neighbors in $ O (\log{n}) $ rounds, (2) every node $ v $ internally checks whether the condition $ \tilde{\dist} (s, v) = \min_{(u, v) \in E} \tilde{\dist} (s, u) + w_G (u, v) $ holds, and (3) the network determines whether the condition holds for all nodes by a global upcast in $ O (D) $ rounds.
If the verification step fails, then the whole SSSP algorithm is repeated.
As $ \mathcal{A} $ is correct with high probability, no repetition is needed with high probability.
Furthermore, by the waiting time bound, there are in expectation only a constant number of repetitions until the verification step succeeds.
\end{proof}

\section{From Positive to Non-Negative Weights}\label{apx:reduction from positive to non-negative weights}

We now argue that for approximate SSSP algorithms it is sufficient to design an algorithm for graphs with positive integer edge weights as such an algorithm can be extended to work for graphs with non-negative integer edge weights.
This reduction seems to be folklore, but we include a formal proof for completeness.
Note that for any given $ \epsilon $, we can always set $ \epsilon' = \tfrac{1}{\lceil 1 / \epsilon \rceil} $ to gurantee that $ \tfrac{1}{\epsilon'} $ is integer, as demanded in the following.
Furthermore, it is usally safe to assume that $ \tfrac{1}{\epsilon} $ is polynomial in $ n $ as otherwise the Bellman-Ford algorithm would compute SSSP faster than approximation algorithms with a round complexity that has polynomial dependence on $ \tfrac{1}{\epsilon} $.

\begin{lemma}\label{lem:positive to non-negative reduction distributed}
Assume there is an algorithm $ \mathcal{A} $ in the CONGEST model that, given a graph $ G' = (V, E, w_{G'}) $ with positive integer edge weights and a source node $ s $, computes a distance estimate $ \dist' (s, \cdot) $ such that
\begin{equation*}
\dist_{G'} (s, v) \leq \dist' (s, v) \leq (1 + \epsilon) \dist^h_{G'} (s, v)
\end{equation*}
for every node $ v \in V $ and for some $ 1 \leq h \leq n - 1 $, where $ \epsilon \geq 0 $ is known and $ \tfrac{1}{\epsilon} $ is integer.
Then there is an algorithm $ \mathcal{B} $ in the CONGEST model that, given a graph $ G = (V, E, w_G) $ with non-negative integer weights in the range $ w_G (e) \in \{ 0, 1, \dots, W \} $ for every edge $ e \in E $ and a source node $ s $, computes a distance estimate $ \tilde{\dist} (s, \cdot) $ such that
\begin{equation*}
\dist_G (s, v) \leq \tilde{\dist} (s, v) \leq  (1 + \epsilon) \dist^h_G (s, v)
\end{equation*}
for every node $ v \in V $ by making one call to algorithm~$ \mathcal{A} $ (on a graph $ G' = (V, E, w_{G'}) $ with integer weights in the range $ w_{G'} (e) \in \{ 1, \dots, O (n W / \epsilon) \} $ for every edge $ e \in E $) and performing internal computation.
\end{lemma}

\begin{proof}
The algorithm $ \mathcal{B} $ first constructs the graph $ G' = (V, E, w_{G'}) $ with weight
\begin{equation*}
w_{G'} (e) := \left(1 + \frac{1}{\epsilon} \right) n w_G (e) + 1
\end{equation*}
for every edge $ e \in E $.
Note that $ G' $ has integer edge weights because we assumed $ \frac{1}{\epsilon} $ to be integer.
Next, it runs the algorithm $ \mathcal{A} $ on $ G' $ to obtain a distance estimate $ \dist' (s, v) $ for every node $ v \in V $.
The algorithm then outputs the distance estimate
\begin{equation*}
\tilde{\dist} (s, v) := \epsilon \left\lfloor \frac{\dist' (s, v)}{(1 + \epsilon) n} \right\rfloor
\end{equation*}
for every node $ v \in V $.
Observe that besides the call to algorithm $ \mathcal{A} $ this algorithm only performs internal computation.
We proceed by proving its correctness.

First, observe that for every path $ \pi $ in the graph we have
\begin{equation*}
w_{G'} (\pi) = \sum_{e \in \pi} \left( \left(1 + \frac{1}{\epsilon} \right) n w_G (e) + 1 \right) = \left(1 + \frac{1}{\epsilon} \right) n \sum_{e \in \pi} w_G (e) + \sum_{e \in \pi} 1 = \left(1 + \frac{1}{\epsilon} \right) n w_G (\pi) + | \pi | \, ,
\end{equation*}
where $ | \pi | $ is the number of edges of $ \pi $.
We therefore have
\begin{equation*}
\dist_{G'} (s, v) \geq  \left(1 + \frac{1}{\epsilon} \right) n \dist_G (s, v)
\end{equation*}
and
\begin{equation*}
\dist^h_{G'} (s, v) \leq \left(1 + \frac{1}{\epsilon} \right) n \dist^h_G (s, v) + h \leq \left(1 + \frac{1}{\epsilon} \right) n \dist^h_G (s, v) + n-1 \, .
\end{equation*}

We prove the lower bound on $ \tilde{\dist} (s, v) $ with the following chain of inequalities:
\begin{equation*}
\tilde{\dist} (s, v) = \epsilon \left\lfloor \frac{\dist' (s, v)}{(1 + \epsilon) n} \right\rfloor \geq \epsilon \left\lfloor \frac{\dist_{G'} (s, v)}{(1 + \epsilon) n} \right\rfloor \geq \epsilon \left\lfloor \frac{\left(1 + \frac{1}{\epsilon} \right) n \dist_G (s, v)}{(1 + \epsilon) n} \right\rfloor = \epsilon \left\lfloor \frac{\dist_G (s, v)}{\epsilon} \right\rfloor = \dist_G (s, v) \, .
\end{equation*}
The last equality follows because $ \tfrac{1}{\epsilon} $ is integer.
We prove the upper bound on $ \tilde{\dist} (s, v) $ with the following chain of inequalities:
\begin{align*}
\tilde{\dist} (s, v) = \epsilon \left\lfloor \frac{\dist' (s, v)}{(1 + \epsilon) n} \right\rfloor &\leq \epsilon \left\lfloor \frac{(1 + \epsilon) \dist^h_{G'} (s, v)}{(1 + \epsilon) n} \right\rfloor \\
&= \epsilon \left\lfloor \frac{\dist^h_{G'} (s, v)}{n} \right\rfloor \\
&\leq \epsilon \left\lfloor \frac{\left(1 + \frac{1}{\epsilon} \right) n \dist^h_G (s, v) + n - 1}{n} \right\rfloor \\
&= \epsilon \left\lfloor \left(1 + \frac{1}{\epsilon} \right) \dist^h_G (s, v) + \frac{n-1}{n} \right\rfloor \\
&= \epsilon \left(1 + \frac{1}{\epsilon} \right) \dist^h_G (s, v) = (1 + \epsilon) \dist^h_G (s, v) \, .
\end{align*}
Again, we exploit that $ \tfrac{1}{\epsilon} $ is integer.
\end{proof}

\printbibliography[heading=bibintoc] 

\end{document}